\documentclass[12pt]{article}
\usepackage{hyperref}
\usepackage{amssymb}
\usepackage{amsmath}
\usepackage{enumerate}
\usepackage{graphicx}
\usepackage{epstopdf}
\usepackage{cite}

\newtheorem{theorem}{Theorem}

\newtheorem{example}{Example}
\newtheorem{lemma}{Lemma}

\newtheorem{remark}{Remark}

\input amssym.def

\newcommand\btd{\raise 2pt \hbox{$\hat\bigtriangledown$}\hskip 1.5pt}
\newcommand\bt{\raise 2pt \hbox{$\bigtriangledown$}\hskip 1.5pt}
\usepackage{indentfirst}
\begin{document}

\title{Tighter monogamy relations for the Tsallis-$q$ and R\'{e}nyi-$\alpha$ entanglement in multiqubit systems}

\author{Rongxia Qi$^{1}$, Yanmin Yang$^{1,2}$ $\thanks{e-mail: ym.yang@kust.edu.cn}$, Jialing Zhang$^{1,2}$, Wei Chen$^{3}$
\\
{\small $^{1}$ Faculty of Science,   Kunming University of Science and Technology, Kunming,  650500, P.R. China}\\
{\small $^{2}$ Research center for Mathematics and Interdisciplinary Sciences,}\\
{\small  Kunming University of Science and Technology, Kunming,  650500, P.R. China}\\
{\small $^{3}$ School of Computer Science and Technology, }\\
{\small Dongguan University of Technology, Dongguan, 523808, P.R. China}\\
}
\date{}
\maketitle

\begin{abstract}
Monogamy relations characterize the distributions of quantum entanglement in multipartite systems. In this work, we present some tighter monogamy relations in terms of  the power of the Tsallis-$q$ and R\'{e}nyi-$\alpha$ entanglement in multipartite systems. We show that these new monogamy relations of multipartite entanglement with tighter lower bounds than the existing ones. Furthermore, three examples are given to illustrate the tightness.

\end{abstract}

\section{Introduction}
Quantum entanglement is an essential feature in terms of quantum mechanics, which distinguishes quantum mechanics from the classical world and plays a very important role in communication, cryptography, and computing. A key property of quantum entanglement is the monogamy relations \cite{bib1,bib2}, which is a quantum systems entanglement with one of the other subsystems limits its  entanglement with the remaining ones, known as the monogamy of entanglement (MoE) \cite{bib2,bib3}. For any tripartite quantum state $\rho_{A\mid BC}$, MoE can be expressed as the following inequality $\mathcal{E}(\rho_{A|BC})\geq \mathcal{E}(\rho_{AB})+\mathcal{E}(\rho_{AC})$, where $\rho_{AB}={\rm tr}_C(\rho_{A|BC})$, $\rho_{AC}={\rm tr}_B(\rho_{A|BC})$, and $\mathcal{E}$ is an quantum entanglement measure. Furthermore, Coffman, Kundu and Wootters  expressed that the squared concurrence also satisfies the monogamy relations in multiqubit states \cite{bib1}. Later  the monogamy relations are widely promoted to other entanglement measures such as entanglement of formation \cite{bib4}, entanglement negativity \cite{bib5}, the Tsallis-$q$ and R\'{e}nyi-$\alpha$  entanglement  \cite{bib6,bib7}. These monogamy relations will help us to have a further understanding of the quantum information theory \cite{bib8}, even black-hole physics \cite{bib9} and condensed-matter physics \cite{bib10}.
In \cite{bib11,bib12}, the authors prove that the  $\eta$th power of Tsallis-$q$  entanglement satisfies monogamy relations for $2\leq q\leq 3$, the power $\eta\geq1$,  the R\'{e}nyi-$\alpha$  entanglement also satisfies  monogamy relations for $\alpha\geq2$, the power $\eta\geq1$, and  $2>\alpha\geq\frac{\sqrt{7}-1}{2}$, the power $\eta\geq2$ .

Our paper is organized as follows. In sec.\ref{sec2}, we review some basic preliminaries  of concurrence, Tsallis-$q$, and R\'{e}nyi-$\alpha$ entanglement. In Sec.\ref{sec3}, we develop a class of monogamy relations in terms of the Tsallis-$q$ entanglement, they are tighter than the results in \cite{bib11}. In Sec.\ref{sec4}, we explore a class of monogamy relations based on the R\'{e}nyi-$\alpha$ entanglement which are tighter than the results in \cite{bib12}. In Sec.\ref{sec5}, we summarize our results.

\section{Basic preliminaries}\label{sec2}
We first recall the definition of concurrence. For a bipartite pure state $|\varphi\rangle_{AB}$, the concurrence  can be defined as \cite{bib17,bib18,bib19}
\begin{equation}
C(|\varphi\rangle_{AB})=\sqrt{2(1-tr\rho_{A}^2)},
\end{equation}
where $\rho_{A}=tr_B(|\varphi\rangle_{AB}\langle\varphi|)$.

For any mixed state $\rho_{AB}$, its concurrence is defined via the convex-roof extension in \cite{bib20}
\begin{equation}
C(\rho_{AB})=min\sum_j p_j C(|\varphi_j\rangle_{AB}),
\end{equation}
where the minimum is taken over all possible pure state decompositions of $\rho_{AB}= \sum\limits_jp_j| \varphi_j\rangle_{AB}\langle\varphi_j|$, and $\sum\limits_jp_j=1$.

It has been proved that the concurrence $C(\rho_{A|B_{1}\cdots B_{N-1}})$ of mixed state $\rho_{A| B_{1}\cdots B_{N-1}}$  has an important property such that \cite{bib21}
\begin{equation}\label{C}
C^{2}(\rho_{A|B_{1}\cdots B_{N-1}})\geq C^{2}(\rho_{A|B_{1}})+C^{2}(\rho_{A|B_2\cdots B_{N-1}})\geq \cdots \geq\sum_{i=1}^{N-1} C^{2}(\rho_{A|B_{i}}),
\end{equation}
where $\rho_{A|B_{i}}=\text{tr}_{B_{1}\cdots B_{i-1}B_{i+1}\cdots B_{N-1}}(\rho_{A|B_{1}\cdots B_{N-1}})$.

Quantum entanglement plays an important role in quantum information. Another well-known quantum entanglements are Tsallis-$q$ entanglement and R\'{e}nyi-$\alpha$ entanglement. For any bipartite pure state $| \varphi\rangle_{AB}$, the Tsallis-$q$ entanglement is defined as \cite{bib23}.
\begin{equation}
T_q(|\varphi\rangle_{AB})=S_{q}(\rho_A)=\frac{1}{q-1}(1-tr\rho_{A}^q),
\end{equation}
where $ q\geq 0$, $q\neq1$, and $\rho_A=tr_B(|\varphi\rangle_{AB}\langle\varphi|)$. When $q$ tends to 1, the Tsallis-$q$ entropy converges to the von Neumann entropy.

For a mixed state $\rho_{AB}$, the Tsallis-$q$ entanglement is defined by its convex-roof extension, which can be expressed as
\begin{equation}
T_q{(\rho_{AB})}=min \sum_i p_i T_q(|\varphi_i\rangle_{AB}),
\end{equation}
where the minimum is taken over all possible pure state decomposition of $\rho_{AB}= \sum\limits_ip_i| \varphi_i\rangle_{AB}\langle\varphi_i|$.

When  $\frac{5-\sqrt{13}}{2}\leq q\leq\frac{5+\sqrt{13}}{2}$, for any  bipartite pure state $| \varphi\rangle_{AB}$, it has been explored that the Tsallis-$q$ entanglement $T_q(|\varphi\rangle_{AB})$ has an analytical formula \cite{bib13},
\begin{equation}
T_q(|\varphi\rangle_{AB})=g_q(C^2(|\varphi\rangle_{AB})),
\end{equation}
where the function $g_q(x)$ is defined as
\begin{equation}
\begin{split}
g_{q}(x)=&\frac{1}{q-1}\left[1-\left(\frac{1+\sqrt{1-x}}{2}\right)^{q}
-\left(\frac{1-\sqrt{1-x}}{2}\right)^{q}\right],
\label{g_q}
\end{split}
\end{equation}
for $0 \leq x \leq 1$, and $g_q(x)$ is an increasing monotonic and convex function in \cite{bib24}. Specially, for $2\leq q\leq3$, the function $g_q(x)$ has an important property \cite{bib23}
\begin{equation}\label{g(x2+y2)}
g_q({x^2+y^2})\geq g_q(x^{2})+g_q(y^{2}).
\end{equation}

When $\frac{5-\sqrt{13}}{2}\leq q\leq\frac{5+\sqrt{13}}{2}$,  for any two-qubit mixed state $\rho$ , the Tsallis-$ q $ entanglement  can be expressed as   $T_q(\rho)=g_q(C^2(\rho))$  \cite{bib24}.

Now, we recall some preliminaries of the R\'{e}nyi-$\alpha$ entanglement. For a bipartite pure state $| \varphi\rangle_{AB}$, the R\'{e}nyi-$\alpha$ entanglement can be  defined as \cite{bib25}
\begin{equation}
E_\alpha(|\varphi\rangle_{AB})=\frac{1}{1-{\alpha}}log_2(tr\rho_A^{\alpha}),
\end{equation}
where  $\alpha>0$, and $\alpha\neq1$, $\rho_A=tr_B(|\varphi\rangle_{AB}\langle\varphi|)$. When $\alpha$ tends to 1, the R\'{e}nyi-$\alpha$ entropy converges to the von Neumann entropy.

For a bipartite mixed state $\rho_{AB}$, the R\'{e}nyi-$\alpha$ entanglement can be defined as
\begin{equation}
E_\alpha(\rho_{AB})=min\sum_i p_i E_\alpha(|\varphi_i\rangle_{AB}),
\end{equation}
where the minimum is taken over all possible pure state decompositions $\{p_i,\varphi_{AB}^i\}$ of $\rho_{AB}$.

When $\alpha\geq \frac{\sqrt{7}-1}{2}$, for any two-qubit state $\rho_{AB}$, the R\'{e}nyi-$\alpha$ entanglement has an analytical formula \cite{bib25,bib26}
\begin{equation}
E_\alpha(\rho_{AB})=f_\alpha(C(\rho_{AB})),
\end{equation}
where $f_\alpha(x)$ can be expressed as
\begin{equation}
f_\alpha(x)=\frac{1}{1-\alpha}log_2\bigg[\bigg(\frac{1-\sqrt{1-x^2}}{2}\bigg)^{\alpha}+\bigg(\frac{1+\sqrt{1-x^2}}{2}\bigg)^{\alpha}\bigg],
\end{equation}
$0\leq x\leq 1$,  and $f_{\alpha}(x)$ is a monotonically increasing convexity function.

For $\alpha\geq2$, the function $f_{\alpha}(x)$ satisfies the following inequality \cite{bib26},
\begin{equation} \label{e1}
f_{\alpha}(\sqrt{x^2+y^2})\geq f_{\alpha}(x)+f_{\alpha}(y).
\end{equation}

For $\frac{\sqrt{7}-1}{2}\leq\alpha<2$, the function $f_{\alpha}(x)$ has an important property such that \cite{bib27}
\begin{equation} \label{e2}
f_\alpha^2(\sqrt{x^2+y^2})\geq f_\alpha^2(x)+ f_\alpha^2(y).
\end{equation}

\section{Tighter monogamy relations in terms of the Tsallis-$q$ entanglement}\label{sec3}
To present the tighter monogamy relations of the Tsallis-$q$  entanglement in  multipartite systems, we introduce three lemmas as follows.

\begin{lemma} \label{lem1}
For $0\leq x\leq 1$ and $\mu\geq1$, we have
\begin{equation}\label{inglemmma1}
\begin{split}
(1+x)^{\mu}&\geq 1+\frac{\mu^2}{\mu+1}x+(2^{\mu}-\frac{\mu^2}{\mu+1}-1)x^{\mu} \\ & \geq1+\frac{\mu}{2}x+(2^{\mu}-\frac{\mu}{2}-1)x^{\mu}\geq1+(2^{\mu}-1)x^{\mu}.
\end{split}
\end{equation}
\end{lemma}

\begin{proof}
If $x=0$, then the inequality is trivial. Otherwise, let $f(\mu,x)=\frac{(1+x)^{\mu}-\frac{\mu^2}{\mu+1}x-1}{x^{\mu}}$, then, $\frac{\partial f}{\partial x}=\frac{\mu x^{\mu-1}[1+\frac{\mu(\mu-1)}{\mu+1}x-(1+x)^{\mu-1}]}{x^{2\mu}}$. When $\mu\geq1$ and $0\leq x\leq1$, it is obvious that $1+\frac{\mu(\mu-1)}{\mu+1}x\leq(1+x)^{\mu-1}$. Thus, $\frac{\partial f}{\partial x}\leq0$, and $f(\mu,x)$ is a decreasing function of $x$, i.e. $ f(\mu,x)\geq f(\mu,1)=2^{\mu}-\frac{\mu^2}{\mu+1}-1$. Consequently, we have $(1+x)^{\mu}\geq 1+\frac{\mu^2}{\mu+1}x+(2^{\mu}-\frac{\mu^2}{\mu+1}-1)x^{\mu}$.
Since $\frac{\mu^2}{\mu+1}\geq \frac{\mu}{2}$, for $0\leq x \leq 1$ and $\mu\geq1$, one gets $ 1+\frac{\mu^2}{\mu+1}x+(2^{\mu}-\frac{\mu^2}{\mu+1}-1)x^{\mu}=1+ \frac{\mu^2}{\mu+1}(x-x^\mu)+(2^{\mu}-1)x^{\mu}\geq 1+\frac{\mu}{2} x+(2^{\mu}-\frac{\mu}{2}-1)x^{\mu}=1+\frac{\mu}{2}(x-x^{\mu})+(2^{\mu}-1)x^{\mu}\geq 1+(2^{\mu}-1)x^{\mu}$.
\end{proof}

\begin{lemma} \label{lem2}
For any $2\leq q\leq 3$, ${\mu}\geq 1$, $g_q(x)$ defined on the domain $D=\{{(x,y)|0\leq x,y\leq 1}\}$, if $x\geq y$, then we have
\begin{equation} \label{inq-lemma2}
\begin{aligned}
g_q^{\mu}({x^2+y^2})& \geq g_q^{\mu}(x^2)+\frac{\mu^2}{\mu+1}g_q^{{\mu}-1}(x^2)g_q(y^2)+(2^{\mu}-\frac{\mu^2}{\mu+1}-1)g_q^{\mu}(y^2).
\end{aligned}
\end{equation}
\end{lemma}

\begin{proof}
For $2\leq q\leq 3$, ${\mu}\geq 1$, according to  inequality (\ref{g(x2+y2)}),   we have
\begin{equation} \label{TSALLIS}
\begin{split}
g_q^{\mu}(x^2+y^2)&\geq (g_q(x^2)+g_q(y^2))^{\mu}\\
&=g_q^{\mu}(x^2)(1+\frac{g_q(y^2)}{g_q(x^2)})^{\mu}\\
&\geq g_q^{\mu}(x^2)+\frac{\mu^2}{\mu+1}g_q^{{\mu}-1}(x^2)g_q(y^2)+(2^{\mu}-\frac{\mu^2}{\mu+1}-1)g_q^{\mu}(y^2),
\end{split}
\end{equation}
where the first inequality is due to inequality (\ref{g(x2+y2)}) and the second inequality is due to Lemma \ref{lem1}.
\end{proof}

\begin{lemma} \label{lem3}
For  any N-qubit mixed state $\rho_{A|B_1\cdots B_{N-1}}$ , we have
\begin{equation} \label{inq-lemma3}
T_q{(\rho_{A|B_1\cdots B_{N-1}})}\geq g_q(C^2(\rho_{A|B_1\cdots B_{N-1}})).
\end{equation}
\end{lemma}

\begin{proof}
Suppose that $\rho_{A|B_1\cdots B_{N-1}}=\sum\limits_i p_i |\varphi_i\rangle _{A|B_1\cdots B_{N-1}}$ is the optimal decomposition for $T_q{(\rho_{A|B_1\cdots B_{N-1}})}$, then we have
\begin{equation}
\begin{split}
T_q{(\rho_{A|B_1\cdots B_{N-1}})} & = \sum_i p_i T_q( |\varphi_i\rangle _{A|B_1\cdots B_{N-1}}) \\
& = \sum_i p_i g_q(C^2(|\varphi_i\rangle _{A|B_1\cdots B_{N-1}}))\\
&\geq g_q (\sum_i p_i C^2(|\varphi_i\rangle _{A|B_1\cdots B_{N-1}}))\\
&\geq g_q ((\sum_i p_i C(|\varphi_i\rangle _{A|B_1\cdots B_{N-1}}))^2)\\
&\geq g_q(C^2(\rho_{A|B_1\cdots B_{N-1}})),
\end{split}
\end{equation}
where the first inequality is due to that $g_q(x)$ is a convex function, the second inequality is due to the Cauchy-Schwarz inequality: $(\sum\limits_{i}a_i^2)(\sum\limits_{i}b_i^2)\geq(\sum\limits_{i}a_ib_i)^2$, with $a_i=\sqrt{p_i}$, and $b_i=\sqrt{p_i}C(|\varphi_i\rangle _{A|B_1\cdots B_{N-1}})$, and the third inequality is due to the minimum property of $C(\rho_{A|B_1\cdots B_{N-1}})$.
\end{proof}

Now, we give the following theorems of the monogamy inequalities in terms of the Tsallis-$q$ entanglement.

\begin{theorem}\label{theorem1}
For any $2\leq q \leq 3$, the power $\eta \geq1$, and $N$-qubit mixed state  $\rho_{A|B_1\cdots B_{N-1}}$, if $C(\rho_{A|B_{i}})\geq C(\rho_{A|B_{i+1}\cdots B_{N-1}}), i=1,2,\cdots,N-2$,  $N>3$,  we have
\begin{equation}\label{inq-theorem1}
\begin{split}
T_{q}^{\eta}(\rho_{A|B_{1}\cdots B_{N-1}}) & \geq \sum_{i=1}^{N-3}h^{i-1}T_{q}^{\eta}(\rho_{A|B_i})
 +h^{N-3}Q_{AB_{N-2}},
\end{split}
\end{equation}
where $h=2^{\eta}-1$, $Q_{AB_{N-2}}=T_{q}^{\eta}(\rho_{A|B_{N-2}})+\frac{\eta^2}{\eta+1}T_{q}^{\eta-1}(\rho_{A| B_{N-2}})T_{q}(\rho_{A|B_{N-1}})+(2^{\eta}-\frac{\eta^2}{\eta+1}-1)T_{q}^{\eta}(\rho_{A|B_{N-1}})$.
\end{theorem}

\begin{proof}
Let $\rho_{A|B_1\cdots B_{N-1}}$ be an $N$-qubit mixed state, from Lemma \ref{lem3} and inequality (\ref{C}), we have
\begin{equation} \label{T1}
\begin{split}
  T_{q}^{\eta}(\rho_{A|B_{1}\cdots B_{N-1}}) & \geq
  g_{q}^{\eta}({C^{2}(\rho_{A|B_{1}})+C^{2}(\rho_{A|B_{2}\cdots B_{N-1}})})\\
  & \geq
  g_{q}^{\eta}(C^2(\rho_{A|B_{1}}))+\frac{\eta^2}{\eta+1} g_{q}^{\eta-1}(C^2(\rho_{A| B_{1}}))g_{q}(C^{2}(\rho_{A|B_{2}\cdots B_{N-1}}))\\
   & \quad+(2^{\eta}-\frac{\eta^2}{\eta+1}-1)g_{q}^{\eta}(C^{2}(\rho_{A|B_{2}\cdots B_{N-1}}))\\
    & \geq g_{q}^{\eta}(C^2(\rho_{A|B_{1}}))+hg_{q}^{\eta}(C^{2}(\rho_{A|B_{2}\cdots B_{N-1}})),
\end{split}
\end{equation}
where the first inequality is due to the monotonically increasing property of the function $g_q(x)$ and inequality (\ref{C}), the second inequality is due to  Lemma \ref{lem2}, and the third inequality is due to the fact that $C^2(\rho_{A|B_{1}})\geq C^2(\rho_{A|B_{2}\cdots B_{N-1}})$.

Similar calculation procedure can be used to the term $g_{q}^{\eta}(C^{2}(\rho_{A|B_{2}\cdots B_{N-1}}))$, by iterative method we can get
\begin{equation} \label{T2}
\begin{split}
& g_{q}^{\eta}(C^{2}(\rho_{A|B_{2}\cdots B_{N-1}}))\\
  \geq  & g_{q}^{\eta}(C^2(\rho_{A|B_{2}}))+hg_{q}^{\eta}({C^{2}(\rho_{A|B_{3}\cdots B_{N-1}})})\geq   \cdots\\
   \geq  &
  g_{q}^{\eta}(C^2(\rho_{A|B_{2}}))+hg_{q}^{\eta}(C^2(\rho_{A|B_{3}}))
  +\cdots+h^{N-5}g_{q}^{\eta}(C^2(\rho_{A|B_{N-3}}))\\
  &\quad+h^{N-4}\bigg\{g_{q}^{\eta}(C^2(\rho_{A|B_{N-2}}))
  +\frac{\eta^2}{\eta+1}g_{q}^{\eta-1}(C^2(\rho_{A|B_{N-2}}))g_{q}(C^2(\rho_{A|B_{N-1}}))\\
  &\quad+(2^{\eta}-\frac{\eta^2}{\eta+1}-1)g_{q}^{\eta}(C^2(\rho_{A|B_{N-1}}))\bigg\}.
\end{split}
\end{equation}
According to the fact $T_q(\rho)=g_q(C^2(\rho))$ for any two qubit  mixed state $\rho$, and combining inequality (\ref{T1}) and (\ref{T2}), we complete the proof.
\end{proof}

\begin{theorem}\label{theorem2}
For any $2\leq q \leq 3$, the power $\eta \geq1$, and $N$-qubit mixed state  $\rho_{A|B_1\cdots B_{N-1}}$,  if $C(\rho_{A|B_{i}})\geq C(\rho_{A|B_{i+1}\cdots B_{N-1}}), i=1, 2, \cdots,m$, $C(\rho_{A|B_{j}})\leq C(\rho_{A|B_{j+1}\cdots B_{N-1}}), j=m+1, m+2, \cdots N-2$, $N>3$, we have
\begin{equation} \label{inq-theorem2}
\begin{split}
 T_{q}^{\eta}(\rho_{A|B_{1}\cdots B_{N-1}}) & \geq \sum_{i=1}^{m}h^{i-1}T_{q}^{\eta}(\rho_{AB_i})+h^{m+1}\sum_{j=m+1}^{N-3}T_{q}^{\eta}(\rho_{AB_j})+h^m{{Q}_{AB_{N-1}}},
\end{split}
\end{equation}
where $h=2^{\eta}-1$, ${{Q}_{AB_{N-1}}}=T_{q}^{\eta}(\rho_{A|B_{N-1}})+\frac{\eta^2}{\eta+1}T_{q}^{\eta-1}(\rho_{A| B_{N-1}})T_{q}(\rho_{A|B_{N-2}})
+(2^{\eta}-\frac{\eta^2}{\eta+1}-1)T_{q}^{\eta}(\rho_{A|B_{N-2}})$.
\end{theorem}

\begin{proof}
For $2\leq q\leq3$, $\eta\geq1$,  we obtain
\begin{equation} \label{t23}
\begin{split}
& T_{q}^{\eta}(\rho_{A|B_{1}\cdots B_{N-1}}) \\
\geq & \sum_{i=1}^{m}h^{i-1}T_{q}^{\eta}(\rho_{AB_i})+h^{m}g_{q}^{\eta}({C^{2}(\rho_{A|B_{m+1}\cdots B_{N-1}})})\\
 \geq & \sum_{i=1}^{m}h^{i-1}T_{q}^{\eta}(\rho_{AB_i})+h^{m}\bigg\{ g_{q}^{\eta}(C^{2}(\rho_{A|B_{m+2}\cdots B_{N-1}}))+(2^{\eta}-\frac{\eta^2}{\eta+1}-1)g_{q}^{\eta}(C^2(\rho_{A|B_{m+1}})) \\
& +\frac{\eta^2}{\eta+1}g_{q}^{\eta-1}({C^{2}(\rho_{A|B_{m+2}\cdots B_{N-1}})})g_{q}(C^2(\rho_{A|B_{m+1}}))\bigg\}\\
\geq & \sum_{i=1}^{m}h^{i-1}T_{q}^{\eta}(\rho_{AB_i})+h^{m+1}g_{q}^{\eta}(C^2(\rho_{A|B_{m+1}}))+h^{m}g_{q}^{\eta}(C^{2}(\rho_{A|B_{m+2}\cdots B_{N-1}}))\geq \cdots\\
  \geq & \sum_{i=1}^{m}h^{i-1}T_{q}^{\eta}(\rho_{AB_i})+h^{m+1}\sum_{j=m+1}^{N-3}g_{q}^{\eta}(C^2(\rho_{AB_j}))+h^m\bigg\{
  (2^{\eta}-\frac{\eta^2}{\eta+1}-1)g_{q}^{\eta}(C^2(\rho_{A|B_{N-2}}))\\
& \quad+g_{q}^{\eta}(C^2(\rho_{A| B_{N-1}}))+\frac{\eta^2}{\eta+1}g_{q}^{\eta-1}(C^2(\rho_{A| B_{N-1}}))g_{q}(C^2(\rho_{A|B_{N-2}}))\bigg\},
\end{split}
\end{equation}
where the first inequality  is due to Theorem \ref{theorem1}, and the second inequality is due to Lemma \ref{lem2} and the fact that $C(\rho_{A| B_{j}})\leq C(\rho_{A|B_{j+1}\cdots B_{N-1}})$ for
 $j=m+1, m+2, \cdots N-2$, $N>3$. According to the denotation of $Q_{AB_{N-1}}$ and combining inequality (\ref{t23}), we obtain inequality (\ref{inq-theorem2}).
\end{proof}

\begin{remark} We consider a particular case of $N=3$. Note that when $2\leq q\leq3$, the power $\eta\geq1$, if $T_q(\rho_{AB_1})\geq T_q(\rho_{AB_2})$, then we get the following result,
\begin{equation}
\begin{split}
T_q^{\eta}(\rho_{A|B_1B_2}) &\geq T_q^{\eta}(\rho_{AB_1})+\frac{\eta^2}{\eta+1}T_q^{\eta-1}(\rho_{AB_1})T_q(\rho_{AB_2})\\
& \quad+(2^{\eta}-\frac{\eta^2}{\eta+1}-1)T_q^{\eta}(\rho_{AB_2}).
\end{split}
\end{equation}
If $T_q(\rho_{AB_1})\leq T_q(\rho_{AB_2})$, then
\begin{equation}
\begin{split}
T_q^{\eta}(\rho_{A|B_1B_2}) &\geq T_q^{\eta}(\rho_{AB_2})+\frac{\eta^2}{\eta+1}T_q^{\eta-1}(\rho_{AB_2})T_q(\rho_{AB_1})\\
& \quad+(2^{\eta}-\frac{\eta^2}{\eta+1}-1)T_q^{\eta}(\rho_{AB_1}).
\end{split}
\end{equation}
\end{remark}

To see the tightness of the Tsallis-$q$ entanglement directly, we give the following example.

\begin{example}
Under local unitary operations, the three-qubit pure state can be written as \cite{bib28,bib29}
\begin{equation} \label{example}
|\psi\rangle_{A|BC}=\lambda_0|000\rangle+\lambda_1e^{i{\varphi}}|100\rangle+\lambda_2|101\rangle
+\lambda_3|110\rangle+\lambda_4|111\rangle,
\end{equation}
where $0\leq \varphi\leq \pi$, $\lambda_i\geq0, i=0,1,2,3,4$, and $\sum_{i=0}^{4}\lambda_i^2 =1$, set $\lambda_0=\frac {\sqrt{5}}{3}$, $\lambda_1=0$, $\lambda_4=0$, $\lambda_2=\frac{\sqrt{3}}{3}$, $\lambda_3=\frac{1}{3}$, $q=2$. From the definition of the Tsallis-$q$  entanglement,
after simple computation, we can get $T_q(\rho_{A| BC})=g_q[(2\lambda_0\sqrt{(\lambda_2)^2+(\lambda_3)^2+(\lambda_4)^2})^2]$, $T_q(\rho_{AB})=g_q[(2{\lambda_0}{\lambda_2})^2]$, and $T_q(\rho_{AC})=g_q[(2{\lambda_0}{\lambda_3})^2]$, then we have $T_q(\rho_{A|BC})=0.49383$, $T_q(\rho_{AB})=0.37037$, and $T_q(\rho_{AC})=0.12346$. Consequently, $T_2^{\eta}(\rho_{A|BC})=(0.49383)^{\eta}\geq T_2^{\eta}(\rho_{AB})+(2^{\eta}-\frac{\eta^2}{\eta+1}-1)T_2^{\eta}(\rho_{AC})+\frac{\eta^2}{\eta+1}T_2^{\eta-1}(\rho_{AB})T_2(\rho_{AC})=
(0.37037)^{\eta}+\frac{\eta^2}{\eta+1}(37037)^{\eta-1}(0.12346)+(2^{\eta}-\frac{\eta^2}{\eta+1}-1)(0.12346)^{\eta}$. While the result in \cite{bib11} is $T_2^{\eta}(\rho_{AB})+(2^{\eta}-1)T_2^{\eta}(\rho_{AC})+\frac{\eta}{2}T_2(\rho_{AC})((T_2^{\eta-1}(\rho_{AB}))-(T_2^{\eta-1}(\rho_{AC})))=(0.37037)^{\eta}+(2^{\eta}-1)(0.12346)^{\eta}+\frac{0.12346{\eta}}{2}((0.37037)^{\eta-1}-(0.12346)^{\eta-1})$.
One can see that our result is  tighter than the ones \cite{bib11} for ${\eta}\geq1$. See FIG. 1.
\begin{figure}[h]
  \centering
\includegraphics[width=11cm,height=5.5cm]{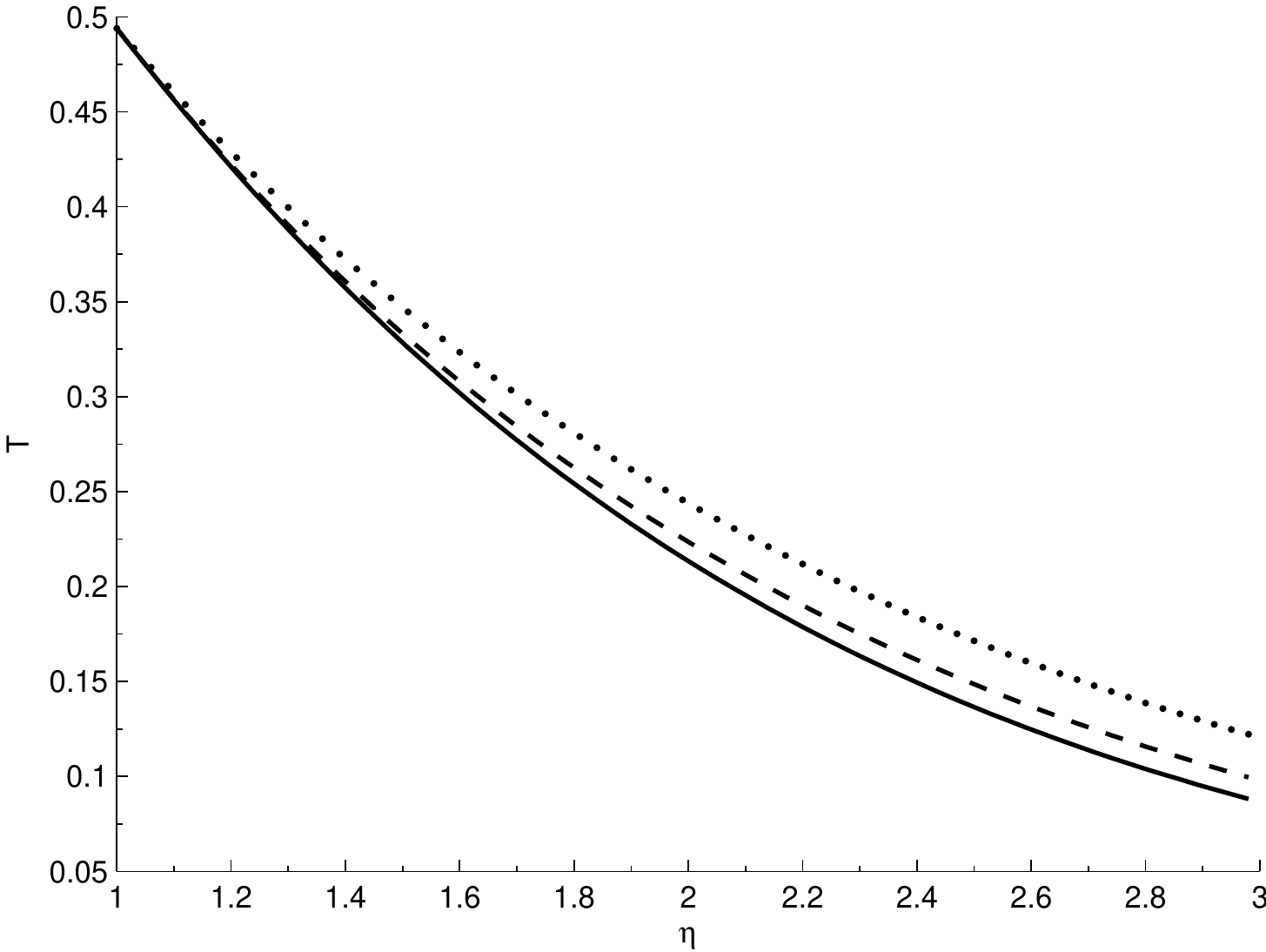}
\caption{The axis T stands the Tsallis-$q$ entanglement of $|\psi\rangle_{A|BC}$, which is a function of $\eta$ $(1\leq \eta\leq3)$. The dotted  line stands the value of $T_2^{\eta}(\rho_{A|BC})$. The dashed line stands the lower bound given by our improved monogamy relations. The  solid black line represents the lower bound  given by  \cite{bib11}.}
\label{1}
\end{figure}
\end{example}

\section{Tighter monogamy relations in terms of the R\'{e}nyi-${\alpha}$  entanglement}\label{sec4}

In order to present the tighter monogamy relations of the R\'{e}nyi-${\alpha}$  entanglement in multiqubit systems, we introduce three lemmas as follows.

\begin{lemma}\label{lem4}
For any N-qubit mixed state $\rho_{A|B_1\cdots B_{N-1}}$, we have
\begin{equation}\label{Exingzhi}
E_{\alpha}{(\rho_{A|B_1\cdots B_{N-1}})}\geq f_{\alpha}(C(\rho_{A|B_1\cdots B_{N-1}})).
\end{equation}
\end{lemma}

\begin{proof}
Suppose that $\rho_{A|B_1\cdots B_{N-1}}=\sum\limits_i p_i |\varphi_i\rangle _{A|B_1\cdots B_{N-1}}$ is the optimal decomposition for $ E_{\alpha}({\rho_{A|B_1\cdots B_{N-1}}})$, then we have
\begin{equation}
\begin{split}
E_{\alpha}({\rho_{A|B_1\cdots B_{N-1}}}) & = \sum_i p_i E_{\alpha}( |\varphi_i\rangle _{A|B_1\cdots B_{N-1}}) \\
& = \sum_i p_i f_{\alpha}(C(|\varphi_i\rangle _{A|B_1\cdots B_{N-1}}))\\
&\geq f_{\alpha} (\sum_i p_i C(|\varphi_i\rangle _{A|B_1\cdots B_{N-1}}))\\
&\geq f_{\alpha}(C(\rho_{A|B_1\cdots B_{N-1}})),
\end{split}
\end{equation}
where the first inequality is due to the convexity of $f_{\alpha}(x)$ and the last inequality follows from the definition of concurrence for mixed state.
\end{proof}

\begin{lemma} \label{lemma5}
For any $\alpha\geq2$, ${\mu}\geq 1$, suppose that the function $f_{\alpha}(x)$ defined on the domain $D=\{{(x,y)|0\leq x,y\leq 1,0\leq x^2+y^2\leq 1}\}$, if $x\geq y$, then we have
\begin{equation} \label{inglemma5}
f_{\alpha}^{\mu}(\sqrt{x^2+y^2})\geq f_{\alpha}^{\mu}(x)+\frac{\mu^2}{\mu+1}f_{\alpha}^{{\mu}-1}(x)f_{\alpha}(y)+(2^{\mu}-\frac{\mu^2}{\mu+1}-1)f_{\alpha}^{\mu}(y).
\end{equation}
\end{lemma}

\begin{proof}
For ${\mu}\geq 1$, and $\alpha\geq2$, we have
\begin{equation}
\begin{split}
f_{\alpha}^{\mu}(\sqrt{x^2+y^2})&\geq (f_{\alpha}(x)+f_{\alpha}(y))^{\mu}\\
&\geq f_{\alpha}^{\mu}(x)+\frac{\mu^2}{\mu+1}f_{\alpha}^{{\mu}-1}(x)f_{\alpha}(y)+(2^{\mu}-\frac{\mu^2}{\mu+1}-1)f_{\alpha}^{\mu}(y),
\end{split}
\end{equation}
where the first inequality is due to inequality (\ref{e1}), and the second inequality is due to Lemma \ref{lem1}.
\end{proof}
\begin{lemma} \label{lemma6}
For any $\frac{\sqrt{7-1}}{2} \leq \alpha < 2$, ${\mu}\geq1$, ${\mu}=\frac{\gamma}{2}$,  the function $f_{\alpha}(x)$ defined on the domain $D=\{{(x,y)|0\leq x,y\leq 1,0\leq x^2+y^2\leq 1}\}$, if $x\geq y$, then we have
\begin{equation} \label{6}
\begin{split}
f_{\alpha}^{\gamma}(\sqrt{x^2+y^2})&\geq f_{\alpha}^{\gamma}(x)+\frac{\mu^2}{\mu+1}f_{\alpha}^{{\gamma}-2}(x)f_{\alpha}^2(y)
+(2^{\mu}-\frac{\mu^2}{\mu+1}-1)f_{\alpha}^{\gamma}(y).
\end{split}
\end{equation}
\end{lemma}

\begin{proof}
For ${\mu}\geq 1$, and ${\mu=\frac{\gamma}{2}}$, we have
\begin{equation}
\begin{split}
f_{\alpha}^{\gamma}(\sqrt{x^2+y^2})&\geq (f_{\alpha}^2(x)+f_{\alpha}^2(y))^{\mu}\\
& \geq f_{\alpha}^{\gamma}(x)+\frac{\mu^2}{\mu+1}f_{\alpha}^{{\gamma}-2}(x)f_{\alpha}^2(y)+(2^{\mu}-\frac{\mu^2}{\mu+1}-1)f_{\alpha}^{\gamma}(y),
\end{split}
\end{equation}
where the first inequality can be assured by inequality (\ref{e2}), and the second inequality is due to Lemma \ref{lem1}.
\end{proof}

Now, we give the following theorems of the tighter monogamy inequality in terms of the R\'{e}nyi-${\alpha}$ entanglement.

\begin{theorem} \label{theorem3}
For any ${\alpha}\geq2 $, the power $\mu \geq1$, and $N$-qubit mixed state  $\rho_{A|B_1\cdots B_{N-1}}$, if $C(\rho_{A|B_{i}})\geq C(\rho_{A|B_{i+1}\cdots B_{N-1}}), i=1,2,\cdots,N-2$, $N>3$, then we have
\begin{equation} \label{inq-theorem3}
\begin{split}
E_{\alpha}^{\mu}(\rho_{A|B_{1}\cdots B_{N-1}})  & \geq \sum_{i=1}^{N-3}h^{i-1}E_{\alpha}^{\mu}(\rho_{A|B_i})
 +h^{N-3}Q_{AB_{N-2}},
\end{split}
\end{equation}
where $h=2^{\mu}-1$, $Q_{AB_{N-2}}=E_{\alpha}^{\mu}(\rho_{A|B_{N-2}})+\frac{\mu^2}{\mu+1}E_{\alpha}^{\mu-1}(\rho_{A| B_{N-2}})E_{\alpha}(\rho_{A|B_{N-1}})
+(2^{\mu}-\frac{\mu^2}{\mu+1}-1)E_{\alpha}^{\mu}(\rho_{A|B_{N-1}})$.
\end{theorem}

\begin{proof}
We consider an $N$-qubit mixed state $\rho_{A|B_1\cdots B_{N-1}}$, from Lemma \ref{lem4}, we have
\begin{equation}
\begin{split}
  E_{\alpha}^{\mu}(\rho_{A|B_{1}\cdots B_{N-1}}) & \geq
  f_{\alpha}^{\mu}\bigg(\sqrt{C^{2}(\rho_{A|B_{1}})+C^{2}(\rho_{A|B_{2}\cdots B_{N-1}})}\bigg)\\
  & \geq
  f_{\alpha}^{\mu}(C(\rho_{A|B_{1}}))+\frac{\mu^2}{\mu+1} f_{\alpha}^{\mu-1}(C(\rho_{A| B_{1}}))f_{\alpha}(C(\rho_{A|B_{2}\cdots B_{N-1}}))\\
   & \quad+(2^{\mu}-\frac{\mu^2}{\mu+1}-1)f_{\alpha}^{\mu}(C(\rho_{A| B_{2}\cdots B_{N-1}}))\\
   & \geq
   f_{\alpha}^{\mu}(C(\rho_{A|B_{1}}))+hf_{\alpha}^{\mu}(C(\rho_{A| B_{2}\cdots B_{N-1}}))\\
  &  \geq
 \cdots\\
  & \geq
  f_{\alpha}^{\mu}(C(\rho_{A|B_{1}}))+hf_{\alpha}^{\mu}(C(\rho_{A|B_{2}}))+\cdots+h^{N-4}f_{\alpha}^{\mu}(C(\rho_{A|B_{N-3}}))\\
  & \quad+h^{N-3}\bigg\{f_{\alpha}^{\mu}(C(\rho_{A|B_{N-2}}))
  +\frac{\mu^2}{\mu+1}f_{\alpha}^{\mu-1}(C(\rho_{A|B_{N-2}}))f_{\alpha}(C(\rho_{A|B_{N-1}}))\\
  & \quad+(2^{\mu}-\frac{\mu^2}{\mu+1}-1)f_{\alpha}^{\mu}(C(\rho_{A|B_{N-1}}))\bigg\},
\end{split}
\end{equation}
where the first inequality is due to the monotonically increasing property of the function $f_{\alpha}(x)$ and inequality (\ref{C}), the second inequality  is due to Lemma \ref{lemma5}, and the third inequality is due to the fact that $C(\rho_{A|B_{i}})\geq C(\rho_{A|B_{i+1}\cdots B_{N-1}})$, $i=1,2,\cdots,N-2$. Then, according to the denotation of $Q_{AB_{N-2}}$ and the definition of the R\'{e}nyi-$\alpha$  entanglement, we complete the proof.
\end{proof}

\begin{theorem} \label{theorem4}
For any $\alpha\geq2$, the power $\mu\geq1$, and $N$-qubit mixed state  $\rho_{A|B_1\cdots B_{N-1}}$, if $C(\rho_{A|B_{i}})\geq C(\rho_{A|B_{i+1}\cdots B_{N-1}})$, $i=1,2,\cdots,m$, $C(\rho_{A|B_{j}})\leq C(\rho_{A|B_{j+1}\cdots B_{N-1}})$, $j=m+1,m+2,\cdots N-2$, $N>3$, then we have
\begin{equation}\label{ingtheorem4}
\begin{split}
 E_{\alpha}^{\mu}(\rho_{A|B_{1}\cdots B_{N-1}}) & \geq \sum_{i=1}^{m}h^{i-1}E_{\alpha}^{\mu}(\rho_{AB_i})+h^{m+1}\sum_{j=m+1}^{N-3}E_{\alpha}^{\mu}(\rho_{AB_j})+h^m{{Q}_{AB_{N-1}}},
\end{split}
\end{equation}
where $h=2^{\mu}-1$, ${{Q}_{AB_{N-1}}}=E_{\alpha}^{\mu}(\rho_{A|B_{N-1}})+\frac{{\mu}^2}{\mu+1}E_{\alpha}^{\mu-1}(\rho_{A| B_{N-1}})E_{\alpha}(\rho_{A|B_{N-2}})
+(2^{\mu}-\frac{{\mu}^2}{\mu+1}-1)E_{\alpha}^{\mu}(\rho_{A|B_{N-2}})$.
\end{theorem}

\begin{proof}
For any $\alpha\geq 2$, $\mu\geq1$, $C(\rho_{A|B_{i}})\geq C(\rho_{A|B_{i+1}\cdots B_{N-1}})$, $i=1,2,\cdots,m$, from Theorem \ref{theorem3}, we know that
\begin{equation} \label{R23}
\begin{split}
 E_{\alpha}^{\mu}(\rho_{A|B_{1}\cdots B_{N-1}}) \geq
  \sum_{i=1}^{m}h^{i-1}E_{\alpha}^{\mu}(\rho_{AB_i})+h^{m}f_{\alpha}^{\mu}(C(\rho_{A|B_{m+1}\cdots B_{N-1}})).
\end{split}
\end{equation}
When  $C(\rho_{A|B_{j}})\leq C(\rho_{A|B_{j+1}\cdots B_{N-1}}), j=m+1, m+2, \cdots N-2$, $N>3$, we get that
\begin{equation} \label{R24}
\begin{split}
f_{\alpha}^{\mu}(C(\rho_{A|B_{m+1}\cdots B_{N-1}})) & \geq f_{\alpha}^{\mu}\bigg(\sqrt{{C^{2}(\rho_{A|B_{m+1}})+C^{2}(\rho_{A|B_{m+2}\cdots B_{N-1}})}}\bigg)\\
&\geq f_{\alpha}^{\mu}({C(\rho_{A|B_{m+2}\cdots B_{N-1}})})+(2^{\mu}-\frac{\mu^2}{\mu+1}-1)f_{\alpha}^{\mu}(C(\rho_{A|B_{m+1}}))\\
&\quad+\frac{\mu^2}{\mu+1}f_{\alpha}^{\mu-1}({C(\rho_{A|B_{m+2}\cdots B_{N-1}})})f_{\alpha}(C(\rho_{A|B_{m+1}}))\\
&\geq f_{\alpha}^{\mu}({C(\rho_{A|B_{m+2}\cdots B_{N-1}})})+hf_{\alpha}^{\mu}(C(\rho_{A|B_{m+1}}))\\
& \geq\cdots\\
  & \geq h\{f_{\alpha}^{\mu}(C(\rho_{A|B_{m+1}}))+\cdots+f_{\alpha}^{\mu}(C(\rho_{A| B_{N-3}}))\}\\
& \quad+\bigg\{f_{\alpha}^{\mu}(C(\rho_{A| B_{N-1}}))+\frac{\mu^2}{\mu+1}f_{\alpha}^{\mu-1}(C(\rho_{A| B_{N-1}}))f_{\alpha}(C(\rho_{A|B_{N-2}}))\\
& \quad +(2^{\mu}-\frac{\mu^2}{\mu+1}-1)f_{\alpha}^{\mu}(C^2(\rho_{A|B_{N-2}}))\bigg\},
\end{split}
\end{equation}
where the first inequality  is due to the monotonically increasing property of the function $f_{\alpha}(x)$  and inequality (\ref{C}), the third inequality is from the  fact that $C(\rho_{A| B_{j}})\leq C(\rho_{A|B_{j+1}\cdots B_{N-1}})$, $j=m+1,m+2,\cdots N-2$, $N>3$. According to the definition of the R\'{e}nyi-$\alpha$  entanglement, and combining inequality (\ref{R23}) and  (\ref{R24}), we obtain inequality (\ref{ingtheorem4}).
\end{proof}

\begin{remark}
 We consider a particular case of $N=3$. Note that when $\alpha\geq2$, the power $\mu\geq1$, if $E_{\alpha}(\rho_{AB_1})\geq E_{\alpha}(\rho_{AB_2})$, then we get the following result,
\begin{equation} \label{f}
\begin{split}
E_{\alpha}^{\mu} (\rho_{A|B_1B_2})&\geq E_{\alpha}^{\mu}(\rho_{AB_1})+\frac{\mu^2}{\mu+1}E_{\alpha}^{\mu-1}(\rho_{AB_1})E_{\alpha}(\rho_{AB_2})\\
& \quad+(2^{\mu}-\frac{\mu^2}{\mu+1}-1)E_{\alpha}^{\mu}(\rho_{AB_2}),
\end{split}
\end{equation}
if $E_{\alpha}(\rho_{AB_1})\leq E_{\alpha}(\rho_{AB_2})$, then
\begin{equation}
\begin{split}
E_{\alpha}^{\mu}(\rho_{A|B_1B_2}) &\geq E_{\alpha}^{\mu}(\rho_{AB_2})+\frac{\mu^2}{\mu+1}E_{\alpha}^{\mu-1}(\rho_{AB_2})E_{\alpha}(\rho_{AB_1})\\
& \quad+(2^{\mu}-\frac{\mu^2}{\mu+1}-1)E_{\alpha}^{\mu}(\rho_{AB_1}).
\end{split}
\end{equation}
\end{remark}

To see the tightness of the R\'{e}nyi-$\alpha$ entanglement directly, we give the following example.
\begin{example}
Let us consider the  state in (\ref{example}) given in Example 1.
Set $\lambda_0=\frac{\sqrt{5}}{3}$, $\lambda_1=\lambda_4=0$, $\lambda_2=\frac{\sqrt{3}}{3}$, $\lambda_3=\frac{1}{3}$, where $\alpha=2$. From  definition of the R\'{e}nyi-$\alpha$ entanglement, after simple computation, we get $E_2({\varphi}_{A|BC})=0.98230$, $E_2({\varphi}_{AB})=0.66742$, $E_2({\varphi}_{AC})=0.19010$, and $E_{2}^{\mu}(\rho_{A|BC})=(0.98230)^{\mu}\geq E_{2}^{\mu}(\rho_{AB})+(2^{\mu}-\frac{\mu^2}{\mu+1}-1)E_{2}^{\mu}(\rho_{AC})+\frac{\mu^2}{\mu+1}E_{2}^{\mu-1}(\rho_{AB})E_{2}(\rho_{AC})$ $=(0.66742)^{\mu}+\frac{\mu^2}{\mu+1}(0.66742)^{\mu-1}(0.19010)+(2^{\mu}-\frac{\mu^2}{\mu+1}-1)(0.19010)^{\mu}$.
While the formula in \cite{bib12} is $ E_{2}^{\mu}(\rho_{AB})+\frac{\mu}{2}E_{2}^{\mu-1}(\rho_{AB})E_{2}(\rho_{AC})+(2^{\mu}-\frac{\mu}{2}-1)E_{2}^{\mu}(\rho_{AC})=(0.66742)^{\mu}+\frac{\mu}{2}(0.66742)^{\mu-1}(0.19010)+(2^{\mu}-\frac{\mu}{2}-1)(0.19010)^{\mu}$.
One can see that our result is tighter than the ones in \cite{bib12} for ${\mu}\geq1$. See FIG. 2.
\begin{figure}[h]
\centering
 \includegraphics[width=11cm,height=5.5cm]{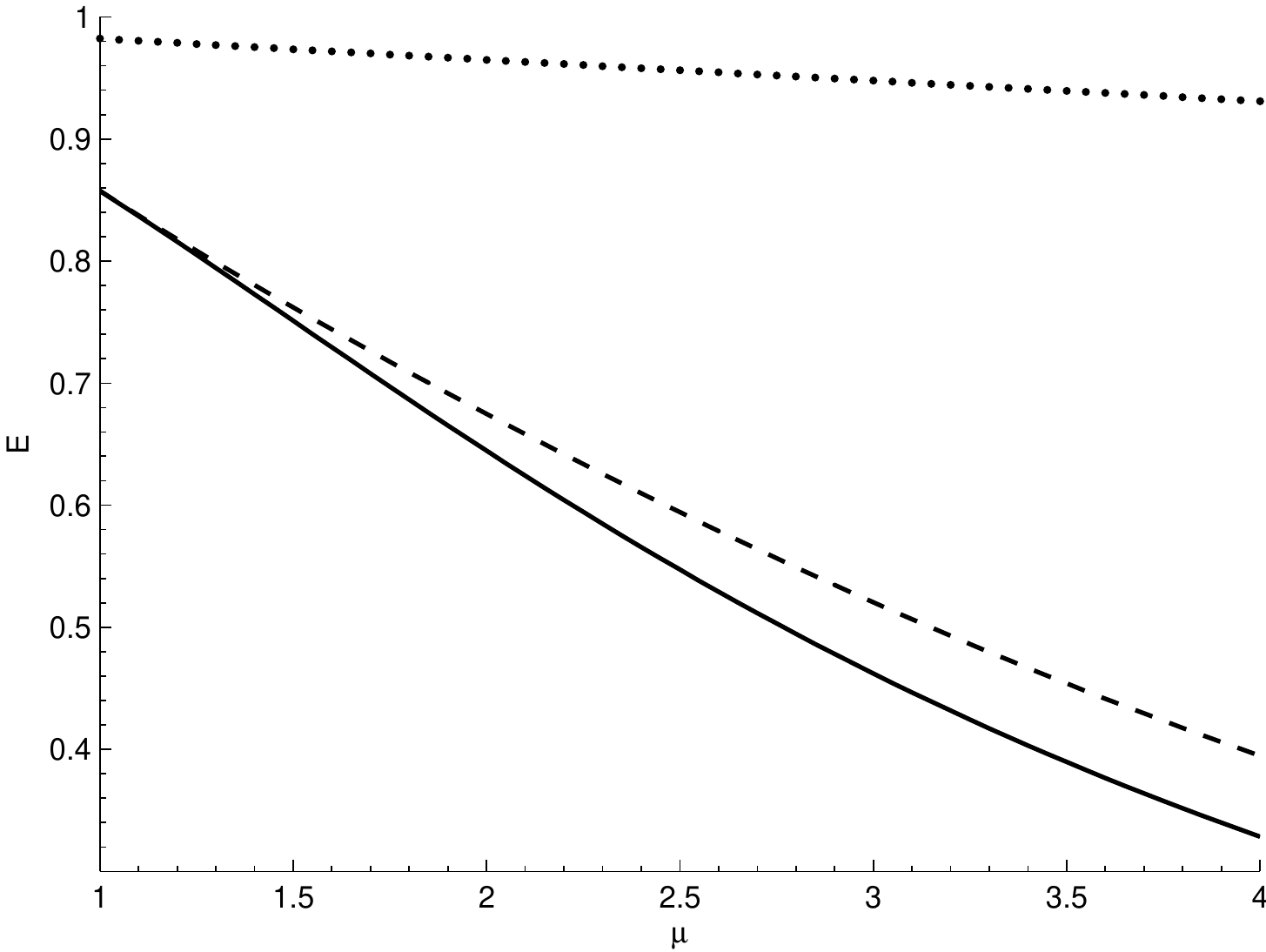}
 \caption{The axis E stands the R\'{e}nyi-$\alpha$ entanglement of $|\psi\rangle_{A|BC}$, which is a function of
  $\mu$ ($1\leq\mu\leq4$). The dotted line stands the value of $E_2^{\mu}(\rho_{A|BC})$.  The dashed line stands the lower bound given by our improved monogamy relations. The solid black line represents the lower bound given by \cite{bib12}.}
\end{figure}
\end{example}

\begin{theorem} \label{theorem5}
For any $\frac{\sqrt{7-1}}{2} \leq \alpha < 2$, the power ${\mu}\geq 1$, ${\mu}=\frac{\gamma}{2}$, and $N$-qubit mixed state  $\rho_{A|B_1\cdots B_{N-1}}$, if $C(\rho_{A|B_{i}})\geq C(\rho_{A|B_{i+1}\cdots B_{N-1}}), i=1,2,\cdots,N-2$, $N>3$, then we have
\begin{equation}\label{inqtheorem5}
\begin{split}
E_{\alpha}^{\gamma}(\rho_{A|B_{1}\cdots B_{N-1}})  & \geq \sum_{i=1}^{N-3}h^{i-1}E_{\alpha}^{\gamma}(\rho_{A|B_i})
 +h^{N-3}Q_{AB_{N-2}},
\end{split}
\end{equation}
where $h=2^{\mu}-1$, $Q_{AB_{N-2}}=E_{\alpha}^{\gamma}(\rho_{A|B_{N-2}})+\frac{\mu^2}{\mu+1}E_{\alpha}^{\gamma-2}(\rho_{A| B_{N-2}})E_{\alpha}(\rho_{A|B_{N-1}})
+(2^{\mu}-\frac{\mu^2}{\mu+1}-1)E_{\alpha}^{\gamma}(\rho_{A|B_{N-1}})$.
\end{theorem}

\begin{proof}
For $\frac{\sqrt{7-1}}{2} \leq \alpha < 2$, ${\mu}\geq 1$, and ${\mu}=\frac{\gamma}{2}$, we consider an $N$-qubit mixed state $\rho_{A|B_1\cdots B_{N-1}}$, from Lemma \ref{lem4}, we have
\begin{equation}
\begin{split}
E_{\alpha}^{\gamma}(\rho_{A|B_{1}\cdots B_{N-1}})& \geq  f_{\alpha}^{2\mu} \bigg(\sqrt{C(\rho_{A|B_{1}})+C^{2}(\rho_{A|B_{2}\cdots B_{N-1}})}\bigg) \\
& \geq
f_{\alpha}^{\gamma}(C(\rho_{A|B_{1}}))+\frac{\mu^2}{\mu+1} f_{\alpha}^{\gamma-2}(C(\rho_{A| B_{1}}))f_{\alpha}^2(C(\rho_{A|B_{2}\cdots B_{N-1}}))\\
& \quad+(2^{\mu}-\frac{\mu^2}{\mu+1}-1)f_{\alpha}^{\gamma}(C(\rho_{A|B_{2}\cdots B_{N-1}})\\
 & \geq
 f_{\alpha}^{\gamma}(C(\rho_{A|B_{1}}))+hf_{\alpha}^{\gamma}(C(\rho_{A|B_{2}\cdot B_{N-1}}))\\
   & \geq
 \cdots\\
  & \geq
  f_{\alpha}^{\gamma}(C(\rho_{A|B_{1}}))+hf_{\alpha}^{\gamma}(C(\rho_{A|B_{2}}))+\cdots+h^{N-4}f_{\alpha}^{\gamma}(C(\rho_{A|B_{N-3}}))\\
  & \quad+h^{N-3}\bigg\{f_{\alpha}^{\gamma}(C(\rho_{A|B_{N-2}}))+\frac{\mu^2}{\mu+1}f_{\alpha}^{\gamma-2}(C(\rho_{A|B_{N-2}}))f_{\alpha}^2(C(\rho_{A|B_{N-1}}))\\
  & \quad+(2^{\mu}-\frac{\mu^2}{\mu+1}-1)f_{\alpha}^{\gamma}(C(\rho_{A|B_{N-1}}))\bigg\},
\end{split}
\end{equation}
where  the first inequality comes from the monotonically increasing property of the function $f_{\alpha}(x)$ and inequality (\ref{C}), the second inequality is due to Lemma \ref{lemma6}, and the third inequality is due to the fact that $C(\rho_{A|B_{i}})\geq C(\rho_{A|B_{i+1}\cdots B_{N-1}})$, $i=1,2,\cdots,N-2$. According to the definition of the R\'{e}nyi-${\alpha}$ entanglement and  the denotation of $Q_{AB_{N-2}}$, we obtain inequality (\ref{inqtheorem5}).
\end{proof}

\begin{theorem} \label{theorem6}
For $\frac{\sqrt{7-1}}{2} \leq \alpha  < 2$, the power ${\mu}\geq 1$, ${\mu}=\frac{\gamma}{2}$, and $N$-qubit mixed state  $\rho_{A|B_1\cdots B_{N-1}}$, if $C(\rho_{A|B_{i}})\geq C(\rho_{A|B_{i+1}\cdots B_{N-1}}), i=1,2,\cdots,m$, $C(\rho_{A|B_{j}})\leq C(\rho_{A|B_{j+1}\cdots B_{N-1}})$, $j=m+1,m+2,\cdots N-2$, $N>3$, then we have
\begin{equation}\label{inqtheorem6}
\begin{split}
E_{\alpha}^{\gamma}(\rho_{A|B_{1}\cdots B_{N-1}}) & \geq \sum_{i=1}^{m}h^{i-1}E_{\alpha}^{\gamma}(\rho_{AB_i})+h^{m+1}\sum_{j=m+1}^{N-3}E_{\alpha}^{\gamma}(\rho_{AB_j})+h^m{{Q}_{AB_{N-1}}},
\end{split}
\end{equation}
where $h=2^{\mu}-1$, ${{Q}_{AB_{N-1}}}=E_{\alpha}^{\gamma}(\rho_{A|B_{N-1}})+\frac{{\mu}^2}{\mu+1}E_{\alpha}^{\gamma-2}(\rho_{A| B_{N-1}})E_{\alpha}(\rho_{A|B_{N-2}})+(2^{\mu}-\frac{{\mu}^2}{\mu+1}-1)E_{\alpha}^{\gamma}(\rho_{A|B_{N-2}})$.
\end{theorem}

\begin{proof}
When $C(\rho_{A| B_{i}})\geq C(\rho_{A|B_{i+1}\cdots B_{N-1}})$, $i=1,2,\cdots,m$, from  Theorem \ref{theorem5}, we have
\begin{equation} \label{gamma23}
\begin{split}
 E_{\alpha}^{\gamma}(\rho_{A|B_{1}\cdots B_{N-1}}) & \geq
  f_{\alpha}^{\gamma}(C(\rho_{A|B_{1}}))+hf_{\alpha}^{\gamma}(C(\rho_{A|B_{2}}))+\cdots+h^{m-1}f_{\alpha}^{\gamma}(C(\rho_{A|B_{m}}))\\
 & \quad +h^{m}f_{\alpha}^{\gamma}(C(\rho_{A|B_{m+1}\cdots B_{N-1}}))\\
 &= \sum_{i=1}^{m}h^{i-1}E_{\alpha}^{\gamma}(\rho_{AB_i})+h^{m}f_{\alpha}^{\gamma}(C(\rho_{A|B_{m+1}\cdots B_{N-1}})).
\end{split}
\end{equation}
When $C(\rho_{A|B_{j}})\leq C(\rho_{A|B_{j+1}\cdots B_{N-1}}), j=m+1, m+2, \cdots N-2$, $N>3$, from Lemma \ref{lemma6}, we get
\begin{equation} \label{gamma24}
\begin{split}
f_{\alpha}^{\gamma}(C(\rho_{A|B_{m+1}\cdots B_{N-1}}))&\geq f_{\alpha}^{2\mu} \bigg(\sqrt{C^{2}(\rho_{A|B_{m+1}})+C^{2}(\rho_{A|B_{m+2}\cdots B_{N-1}})}\bigg)\\
&\geq f_{\alpha}^{\gamma}(C(\rho_{A|B_{m+2}\cdots B_{N-1}}))+(2^{\mu}-\frac{\mu^2}{\mu+1}-1)f_{\alpha}^{\gamma}(C(\rho_{A|B_{m+1}}))\\
 & \quad +\frac{\mu^2}{\mu+1}f_{\alpha}^{\gamma-2}(C(\rho_{A|B_{m+2}\cdots B_{N-1}}))f_{\alpha}^2(C(\rho_{A|B_{m+1}}))\\
  & \geq  f_{\alpha}^{\gamma}(C(\rho_{A|B_{m+2}\cdots B_{N-1}}))+hf_{\alpha}^{\gamma}(C(\rho_{A|B_{m+1}}))\\
  &\geq \cdots\\
  &\geq hf_{\alpha}^{\gamma}(C(\rho_{A| B_{m+1}}))+\cdots+hf_{\alpha}^{\gamma}(C(\rho_{A| B_{N-3}}))\\
  & \quad+\bigg\{f_{\alpha}^{\gamma}(C(\rho_{A| B_{N-1}}))+\frac{\mu^2}{\mu+1}f_{\alpha}^{\gamma-2}(C(\rho_{A|B_{N-1}}))f_{\alpha}^2(C(\rho_{A|B_{N-2}}))\\
  &\quad+(2^{\mu}-\frac{\mu^2}{\mu+1}-1)f_{\alpha}^{\gamma}(C(\rho_{A|B_{N-2}}))\bigg\},
\end{split}
\end{equation}
where  the first inequality comes from the monotonically increasing property of the function $f_{\alpha}(x)$ and inequality (\ref{C}), the second inequality is due to Lemma \ref{lemma6}, and the third inequality is due to the fact that $C(\rho_{A|B_{j}})\leq C(\rho_{A|B_{j+1}\cdots B_{N-1}})$, $j=m+1, m+2, \cdots N-2$, $N>3$.  According to the denotation of $Q_{AB_{N-1}}$ and combining inequality (\ref{gamma23}) and (\ref{gamma24}), we complete the proof.
\end{proof}

\begin{remark}
 We consider a particular case of $N=3$. Note that when $\frac{\sqrt{7}-1}{2}\leq\alpha<2$, the power $\mu\geq1$ and $\mu=\frac{\gamma}{2}$, if $E_{\alpha}(\rho_{AB_1})\geq E_{\alpha}(\rho_{AB_2})$, then we get the following result,
\begin{equation}
\begin{split}
E_{\alpha}^{\gamma}(\rho_{A|B_1B_2})& \geq E_{\alpha}^{\gamma}(\rho_{AB_1})+\frac{\mu^2}{\mu+1}E_{\alpha}^{\gamma-2}(\rho_{AB_1})E_{\alpha}^{2}(\rho_{AB_2})\\
& \quad+\left.(2^{\mu}-\frac{\mu^2}{\mu+1}-1)E_{\alpha}^{\gamma}(\rho_{AB_2})\right.,
\end{split}
\end{equation}
if $E_{\alpha}(\rho_{AB_1})\leq E_{\alpha}(\rho_{AB_2})$, then
\begin{equation}
\begin{split}
E_{\alpha}^{\gamma}(\rho_{A|B_1B_2})& \geq E_{\alpha}^{\gamma}(\rho_{AB_2})+\frac{\mu^2}{\mu+1}E_{\alpha}^{\gamma-2}(\rho_{AB_2})E_{\alpha}^{2}(\rho_{AB_1})\\
& \quad+\left.(2^{\mu}-\frac{\mu^2}{\mu+1}-1)E_{\alpha}^{\gamma}(\rho_{AB_1})\right..
\end{split}
\end{equation}
\end{remark}

To see the tightness of the R\'{e}nyi-${\alpha}$ entanglement directly, we give the following example.

\begin{example}
Let us consider the  state in (\ref{example}) given in Example 1. Suppose that $\lambda_0=\frac{\sqrt{5}}{3}$, $\lambda_1=\lambda_4=0$, $\lambda_2=\frac{\sqrt{3}}{3}$, $\lambda_3=\frac{1}{3}$, and $\alpha=\frac{\sqrt{7}-1}{2}$. From  definition of the R\'{e}nyi-$\alpha$ entanglement, we have $E_{\alpha}(|\psi\rangle_{A|BC})=0.99265$,  $E_{\alpha}(|\psi\rangle_{AB})=0.83477$, $E_{\alpha}(|\psi\rangle_{AC})=0.41466$, and $E_{\alpha}^{\gamma}(\rho_{A|BC})=(0.99265)^{\gamma}\geq E_{\alpha}^{\gamma}(\rho_{AB})+\frac{\gamma^2}{4+2\gamma}E_{\alpha}^{\gamma-2}(\rho_{AB})E_{\alpha}^{2}(\rho_{AC}) +\left.(2^{\frac{\gamma}{2}}-\frac{\gamma^2}{4+2\gamma}-1)E_{\alpha}^{\gamma}(\rho_{AC})\right.=(0.83477)^{\gamma}+\frac{\gamma^2}{4+2\gamma}(0.83477)^{\gamma-2}(0.41466)^{2}$
$+\left.(2^{\frac{\gamma}{2}}-\frac{\gamma^2}{4+2\gamma}-1)(0.41466)^{\gamma}\right..$
While the formula in \cite{bib12} is $ E_{\alpha}^{\gamma}(\rho_{AB})+\frac{\gamma}{4}E_{\alpha}^{\gamma-2}(\rho_{AB})E_{\alpha}^{2}(\rho_{AC})+(2^{\frac{\gamma}{2}}-\frac{\gamma}{4}-1)E_{\alpha}^{\gamma}(\rho_{AC})=(0.83477)^{\gamma}+\frac{\gamma}{4}(0.83477)^{\gamma-2}\\(0.41466)^{2}+(2^{\frac{\gamma}{2}}-\frac{\gamma}{4}-1)(0.41466)^{\gamma}.$ One can see that our result is tighter than the result in \cite{bib12} for $\mu=\frac{\gamma}{2}$, $\gamma\geq2$. See FIG. 3.
\begin{figure}[h]
\centering
\includegraphics[width=11cm,height=5.5cm]{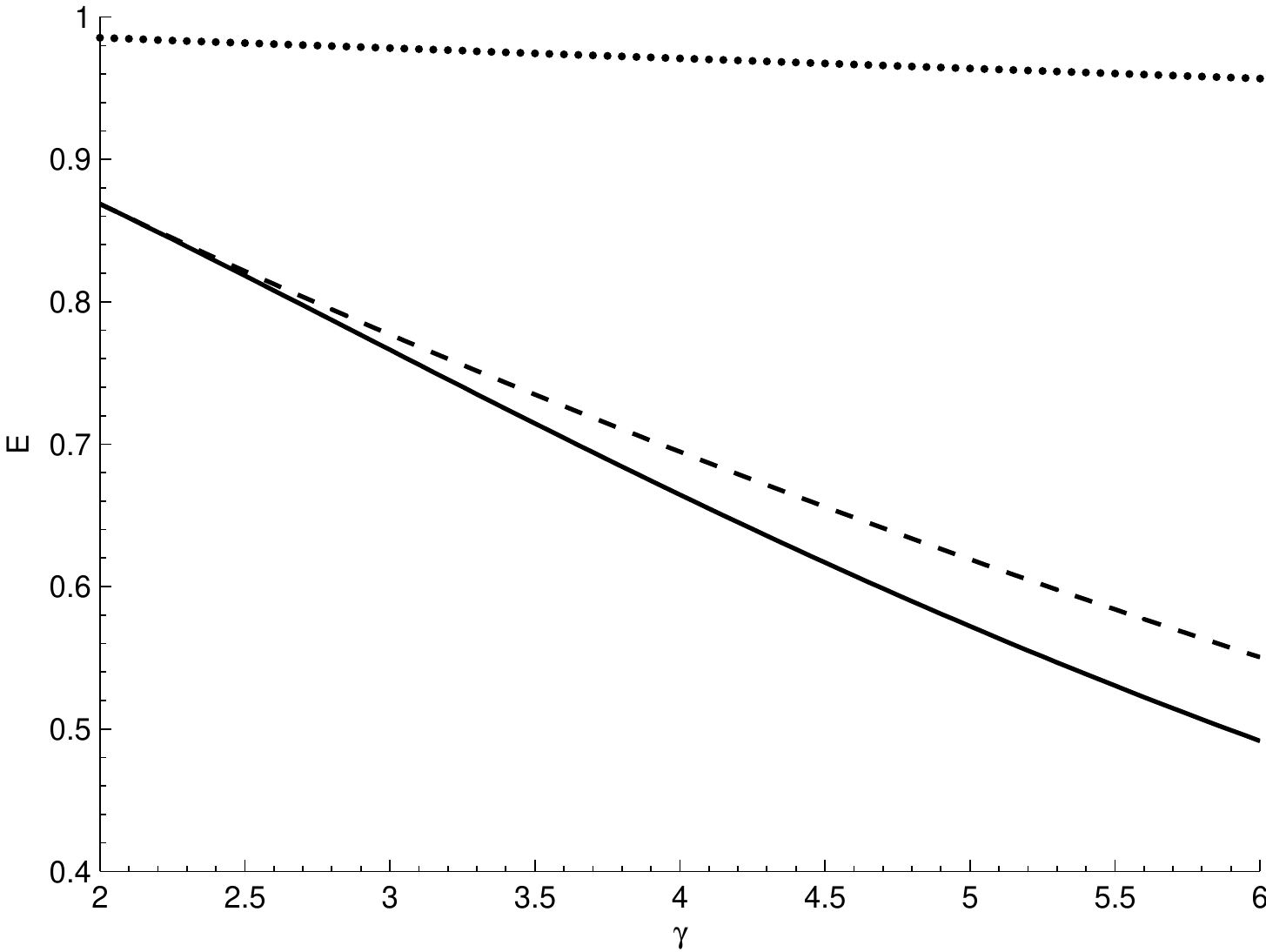}
\caption{The axis E stands the R\'{e}nyi-$\alpha$ entanglement of $|\psi\rangle_{A|BC}$, which is a function of $\gamma$ ($2\leq\gamma\leq6$). The dotted line stands the value of $E_{\frac{\sqrt{7}-1}{2}}^{\gamma}(\rho_{A|BC})$.   The dashed line stands the lower bound given by our improved monogamy relations. The solid black line represents the lower bound given by \cite{bib12}.}
\label{3}
\end{figure}
\end{example}

\section{Conclusion}\label{sec5}
Multipartite entanglement can be regarded as a fundamental problem in the theory of quantum entanglement. Our results may contribute to a fuller understanding of the Tsallis-$q$ and R\'{e}nyi-$\alpha$  entanglement in  multipartite systems. In this paper, we have explored some tighter monogamy relations  in terms of  $\eta$th power of the Tsallis-$q$ entanglement $T_{q}^{\eta}(\rho_{A|B_{1}\cdots B_{N-1}})$ ( $\eta\geq1$, $2\leq q\leq 3$) and  the R\'{e}nyi-$\alpha$ entanglement $E_{\alpha}^{\mu}(\rho_{A|B_{1}\cdots B_{N-1}})$ ($\mu\geq1$, $\alpha\geq2$ ) and  $E_{\alpha}^{\gamma}(\rho_{A|B_{1}\cdots B_{N-1}})$ ($\gamma\geq2$, $\frac{\sqrt{7}-1}{2}\leq\alpha<2$ ).
We show that these new monogamy relations of multiparty entanglement have   larger lower bounds and are tighter than the existing results \cite{bib11,bib12}. Our  approach may also be applied to the study of monogamy properties  related to other quantum correlations.

\section*{Acknowledgments}
This work is  supported by the Yunnan Provincial Research Foundation for Basic Research, China (Grant No. 202001AU070041), the Research Foundation of Education Bureau of Yunnan Province, China (Grant No. 2021J0054), the Basic and Applied Basic Research Funding Program of Guangdong Province (Grant No. 2019A1515111097), the Natural Science Foundation of Kunming University of Science and Technology (Grant No. KKZ3202007036, KKZ3202007049).

\section*{Statements and Declarations}

{\bf Funding}

This work is  supported by the Yunnan Provincial Research Foundation for Basic Research, China (Grant No. 202001AU070041), the Research Foundation of Education Bureau of Yunnan Province, China (Grant No. 2021J0054), the Basic and Applied Basic Research Funding Program of Guangdong Province (Grant No. 2019A1515111097), the Natural Science Foundation of Kunming University of Science and Technology (Grant No. KKZ3202007036, KKZ3202007049).

{\bf  Competing Interests}

The authors have no relevant financial or non-financial interests to disclose.

{\bf Author Contributions}

All authors contributed to the study conception and design. Material preparation, data collection and analysis were performed by Rongxia Qi and Yanmin Yang.  The first draft and critically revised the manuscript were done by all authors. All authors read and approved the final manuscript.


\begin{thebibliography}{99}
\bibitem{bib1}Coffman, V., Kundu, J.,  Wootters, W.K.: Distributed entanglement. Physical Review A 61,052306(2000)
\bibitem{bib2}Terhal, B.M.: Is entanglement monogamous? IBM Journal of Research and Development 48,71(2004)
\bibitem{bib3}Kim, J.S., Gour, G., Sanders, B.C.: Limitations to sharing entanglement. Contemporary Physics 53,417(2012)
\bibitem{bib4}Bai, Y.K., Xu, Y.F., Wang, Z.D.: General Monogamy Relation for the Entanglement of Formation in Multiqubit Systems. Physical review letters 113,100503(2014)
\bibitem{bib5}Ou, Y.C., Fan, H.: Monogamy inequality in terms of negativity for three-qubit states. Physical Review A 75,062308(2007)
\bibitem{bib6}Halasz, G.B.,  Hamma, A.: Topological R\'{e}nyi Entropy after a Quantum Quench. Physical review letters 110,170605(2013)
\bibitem{bib7}Hamma, A., Cincio, L., Santra, S.,  Zanardi,  P.,  Amico, L.: Local Response of Topological Order to an External Perturbation. Physical review letters 110,100503(2013)
\bibitem{bib8}Seevinck, M.P.: Monogamy of correlations versus monogamy of entanglement. Quantum Information Processing 9,273(2010)
\bibitem{bib9}Verlinde, E., Verlinde, H.: Black hole entanglement and quantum error correction. Journal of High Energy Physics 2013,107(2013)
\bibitem{bib10}Ma, X., Dakic,  B., Naylor, W.,  Zeilinger, A.,  Walther, P.: Quantum simulation of the wavefunction to probe frustrated Heisenberg spin systems. Nature Physics 7,399(2011)
\bibitem{bib11}Zhang, J.B., Jin, Z.X., Fei, S.M., Wang, Z.X.: Enhanced Monogamy Relations in Multiqubit Systems. International Journal of Theoretical Physics 59,3449(2020)
\bibitem{bib12}Gao, L.M., Yan, F.L., Gao, T.: Tighter monogamy relations of multiqubit entanglement in terms of R\'{e}nyi-$\alpha$ entanglement. Communications in Theoretical Physics 72,085102(2020)
\bibitem{bib17}Rungta, P.,  Bu\v{z}ek, V., Caves, C.M., Hillery, M., Milburn, G.J.: Universal state inversion and concurrence in arbitrary dimensions. Physical Review A 64,042315(2001)
\bibitem{bib18}Uhlmann, A.: Fidelity and concurrence of conjugated states. Physical Review A 62,032307(2000)
\bibitem{bib19}Albeverio, S., Fei, S.M.: A note on invariants and entanglements. Journal of Optics B: Quantum and Semiclassical Optics 3,223(2001)
\bibitem{bib20}Uhlmann, A.: Roofs and Convexity. Entropy 12,1799(2010)
\bibitem{bib21}Osborne, T.J., Verstraete, F.: General Monogamy Inequality for Bipartite Qubit Entanglement.  Physical review letters 96,220503(2006)
\bibitem{bib23}Kim, J.S.: Tsallis entropy and entanglement constraints in multiqubit systems. Physical Review A 81,062328(2010)
\bibitem{bib13}Yuan, G.M., Song, W., Yang, M., Li, D.C., Zhao, J.L., Cao, Z.L.: Monogamy relation of multi-qubit systems for squared Tsallis-q entanglement. Scientific reports 6,28719(2016)
\bibitem{bib24}Luo, Y., Tian, T., Shao, L.H., Li, Y.: General monogamy of Tsallis q-entropy entanglement in multiqubit systems. Physical Review A 93,062340(2016)
\bibitem{bib25}Wang, Y.X., Mu, L.Z., Vedral, V., Fan, H.: Entanglement R\'{e}nyi entropy. Physical Review A 93,022324(2016)
\bibitem{bib26}Kim, J.S., Sanders, B.C.: Monogamy of multi-qubit entanglement using R\'{e}nyi entropy. Journal of Physics A: Mathematical and Theoretical 43,445305(2010)
\bibitem{bib27}Song, W., Bai, Y.K., Yang, M., Cao, Z.L.: General monogamy relation of multiqubit systems in terms of squared R\'{e}nyi-$\alpha$ entanglement.  Physical Review A 93,022306(2016)
\bibitem{bib28}Ac\'{i}n,  A., Andrianov, A., Costa, L.,  Jan\'{e}, E.,  Latorre, J.I.,  Tarrach, R.: Generalized Schmidt Decomposition and Classification of Three-Quantum-Bit States. Physical Review Letters 85,1560(2000)
\bibitem{bib29}Gao, X.H., Fei, S.M.: Estimation of concurrence for multipartite mixed states. The European Physical Journal Special Topics 159,71(2008)
\end{thebibliography}
\end{document}